\documentclass[letterpaper,10pt,conference]{ieeeconf}

\IEEEoverridecommandlockouts
\usepackage{amsmath,amssymb,amsfonts}
\usepackage{booktabs}
\usepackage{cite}
\usepackage{graphicx}
\usepackage{xcolor}
\usepackage{algorithm}
\usepackage{algpseudocode}
\usepackage{mathtools}
\makeatletter
\let\NAT@parse\undefined
\makeatother
\usepackage[hidelinks]{hyperref}
\usepackage{caption}
\usepackage{subcaption}

\newtheorem{theorem}{Theorem}
\newtheorem{remark}{Remark}
\newtheorem{assumption}{Assumption}

\newcommand{\R}{\mathbb{R}}

\title{Accelerating Branch MPC with Two-Level Parallel Direct Solves on GPUs}

\author{Fenglong Song$^{1}$, Luyao Zhang$^2$, Liang Wu$^{3}$, J\'an Drgo\v na$^3$, and Colin N. Jones$^{1}$%
\thanks{This work is supported by NCCR Automation, a National Centre of Competence in Research, funded by the Swiss National Science Foundation (grant number 51NF40\_225155). This work is also partially supported by the U.S. DOE, Office of Science, ASCR program under the Scientific Discovery through Advanced Computing (SciDAC) Institute “LEADS: LEarning-Accelerated Domain Science” and by the Ralph O’Connor Sustainable Energy Institute (ROSEI) at Johns Hopkins University.}%
\thanks{$^{1}$Automatic Control Laboratory, \'{E}cole polytechnique f\'{e}d\'{e}rale de Lausanne (EPFL).
{\tt\small \{fenglong.song\}@epfl.ch}}%
\thanks{$^{2}$Delft Center for Systems and Control, Delft University of Technology.
}%
\thanks{$^{3}$Department of Civil and Systems
Engineering, Johns Hopkins University.
}%
}

\begin{document}

\maketitle
\thispagestyle{empty}
\pagestyle{empty}

\begin{abstract}
Branch model predictive control optimizes multiple future trajectories coupled through shared decisions, with computational demands increasing as the number of scenarios and prediction horizon grow. We present a GPU-accelerated direct linear solver for branch MPC formulations in which all trajectories share a single root decision node and evolve independently thereafter. By operating at the linear-algebra level, the solver provides a reusable backend for multiple optimization algorithms whose reduced systems have the required symmetric positive-definite structure. The solver exploits two levels of parallelism: across scenarios and along each prediction horizon. A tailored variable ordering enables horizon-parallel Cholesky factorization while preserving a single root–tail coupling block per scenario in the factor. Numerical experiments demonstrate substantial speedups over state-of-the-art sparse direct solvers, achieving factorization speedups of up to 6.0$\times$ over \texttt{cuDSS} and 27.6$\times$ over eight-thread \texttt{PARDISO}, with triangular solve speedups of up to 3.5$\times$ and 15.8$\times$, respectively.

% Branch model predictive control represents alternative future evolutions as trajectory branches coupled through partially shared decisions. However, its computational cost grows rapidly with the horizon length and number of scenarios, making real-time execution challenging. GPUs provide the parallel computing resources needed to alleviate this bottleneck. While previous work has exploited parallelism across scenarios and prediction stages at the algorithmic level, we address the problem directly at the linear algebra level, producing a backend that can be readily integrated into various solver frameworks. We develop a GPU direct solver that parallelizes both across scenarios and along the prediction horizon, achieving $\mathcal{O}(\log N)$ parallel time given sufficient computing resources, where $N$ is the horizon length. Numerical experiments demonstrate substantial speedups over state-of-the-art sparse direct solvers.

\end{abstract}

%%%%%%%%%%%%%%%%%%%%%%%%%%%%%%%%%%%%%%%%%%%%%%%%%%%%%%%%%%%%%%%%%%%%%%%%%%%%%%%%
% ============================== SECTION: Introduction ==============================
%%%%%%%%%%%%%%%%%%%%%%%%%%%%%%%%%%%%%%%%%%%%%%%%%%%%%%%%%%%%%%%%%%%%%%%%%%%%%%%%
\section{Introduction}
\label{sec:introduction}

% \subsection{Branch Model Predictive Control}
Model predictive control (MPC) selects control actions using predictions of future system behavior. Under uncertainty, however, a single predicted trajectory may not represent the range of disturbances or operating modes that the system may encounter. Branch MPC addresses this limitation by optimizing multiple candidate future trajectories whose decisions can diverge as additional information becomes
available~\cite{oliveira2023interfaction,chen2022interactive}, preparing different responses to possible future evolutions. However, its computational cost grows with both the number of trajectories and the prediction horizon.

% This paper considers branch MPC primarily in the setting of scenario-based MPC~\cite{bernardini2009scenario,calafiore2013scenario}, where sampled realizations of uncertain parameters or disturbances generate a collection of scenario trajectories. Representing complex uncertainty distributions and low-probability events can require hundreds or even thousands of scenarios~\cite{degroot2025scenario, sampathirao2018gpu}, and these scenarios are usually highly regular, i.e., have uniform horizon length and variable sizes. Such large and relatively regular collections of trajectories provide substantial parallelism and great potential for Graphics Processing Units (GPUs) to take advantage of.
% Contingency MPC~\cite{zheng2026contingencyreview, alasterda2019contigency, alsterda2021contingency} is another setting in which branching trajectories are used to represent alternative future events, but typically involves fewer and less regular branches and is not our primary target.

This paper proposes a computational approach that applies directly to the broad class of branch MPC problems that appears under different names throughout the MPC literature, such as min-max feedback MPC over disturbance trees~\cite{scokaert1998minmax}, scenario-tree stochastic
MPC~\cite{bernardini2009scenario,mesbah2016stochastic}, multi-stage robust MPC~\cite{lucia2013multistage}, risk-averse MPC in which the expectation over the tree is replaced by a coherent risk measure~\cite{sopasakis2019riskaverse}, contingency MPC for automated
driving~\cite{alasterda2019contigency,alsterda2021contingency,zheng2026contingencyreview}, and branch MPC for interactive, multi-modal motion planning~\cite{chen2022interactive}. These formulations differ in how the tree is generated, but they all produce optimization problems, and hence linear systems, with the same tree-induced sparsity.

Specifically, we consider a formulation in which all trajectories share one root decision node and evolve independently thereafter 
% (called \textit{scenario fan} in some literature~\cite{Sampathirao2015distributed, Patrinos2011stochastic})
, as illustrated in Fig.~\ref{fig:branch_mpc}.
This structure exposes two complementary sources of parallelism: different scenario tails can be processed concurrently, while operations along the horizon can also be organized into parallel levels.
Such large and relatively regular collections of trajectories provide substantial parallelism and great potential for Graphics Processing Units (GPUs) to take advantage of.

\begin{figure}[!t]
\centering
\includegraphics[width=\columnwidth]{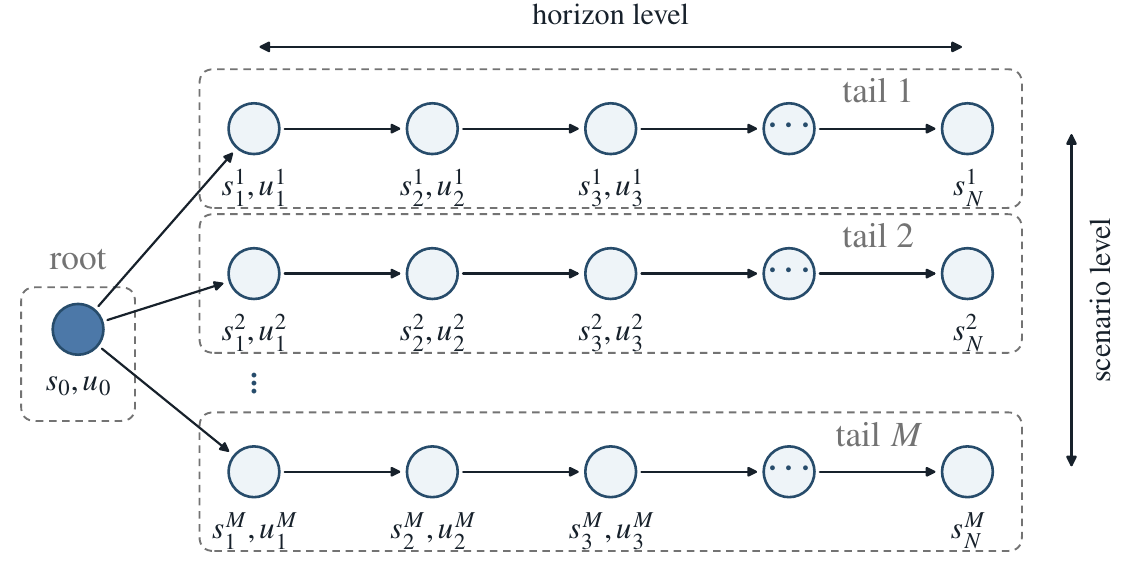}
\caption{Illustration for the branch MPC we consider in this work, where $s$ stands for states and $u$ stands for inputs. All $M$ scenario tails share one complete root decision node $(s_0,u_0)$ and evolve independently
afterwards.  
% The two annotated directions are the parallelism levels
% exploited by the proposed solver.
}
\label{fig:branch_mpc}
\end{figure}

% \subsection{Related Work and Contributions}
The sparsity structure of the scenario tree has been exploited in related work. The Riccati recursion can be used to process independent subtrees toward the root, requiring work linear in the number of nodes~\cite{frison2017tree,frison2020hpipm}. Dual-Newton methods provide a complementary decomposition by operating directly on tree-structured quadratic programs and permit parallel branch processing without a scenario-wise reformulation~\cite{kouzoupis2019treeqp}. 
% \FS{should mention~\cite{sampathirao2018gpu, Sampathirao2015distributed}} 
Although these approaches expose concurrency across branches, their temporal recursions remain sequential along each branch, limiting the available parallelism for long prediction horizons.

This limitation motivates the use of temporal parallelization methods developed for optimal control. Associative scans enable parallel linear-quadratic control along the prediction
horizon~\cite{sarkka2023temporal,amatuccisousa26ral}. With sufficient computational resources, it can reduce the time complexity from $\mathcal{O}(N)$ to $\mathcal{O}(\log N)$, where $N$ is the horizon length. Related structure-exploiting factorization methods organize the computation into a hierarchy of parallel subproblems, likewise achieving logarithmic time complexity~\cite{nielsen2015parallel,pas2026cyqlone,schwan2026socu,song2025parallelkktsolverpiqp}. These techniques provide building blocks for parallelizing the computations within each trajectory branch.

Recent work exploits both directions of parallelism through solver-specific recursions.
The method in~\cite{zhang2026parallelbranch} embeds associative scans within an augmented Lagrangian framework, while \cite{sousapinto2026rakecompress} applies rake--compress Riccati recursion to general scenario trees.  
Our approach targets a different computation layer. We exploit the two-level parallelism directly \emph{at the linear-algebra level}, performing parallel Cholesky factorization and triangular solves when solving the linear system arising from the optimization algorithms. It can therefore serve as a backend for multiple optimization methods, such as interior-point methods (IPMs) and the
alternating direction method of multipliers (ADMM), whenever they produce linear systems with the same target structure (see Appendix~\ref{app:equivalence_lin_sys_sparsity}).

% \subsection{Contributions}

Our contributions are:
\begin{itemize}
  % \item We show that immediate-branching branch MPC and delayed-branching contingency MPC after shared-prefix condensation produce a common root-coupled block-tridiagonal linear-system structure.
  % \item We derive an exact direct factorization and solve that combines concurrent processing of trajectory branches with parallelism along each horizon.
  % \item We implement a latency-oriented backend for numerical factorization and repeated single-right-hand-side solves, and define an evaluation that isolates both levels of parallelism on Branch-MPC-derived and controlled synthetic systems.
  \item We develop a direct factorization and triangular solve procedure for root-coupled branch MPC systems that combines parallelism across scenarios and along each prediction horizon. 
  % A tailored variable ordering preserves a single root–tail coupling block per scenario in the Cholesky factor, enabling compact root updates and coupling operations.
  \item We implement\footnote{Link to code: \href{https://github.com/PREDICT-EPFL/branch-mpc-linsys-solver}{github.com/PREDICT-EPFL/branch-mpc-linsys-solver}} the method as a GPU linear-algebra backend, which can be integrated into multiple optimization algorithms whose reduced systems have the required symmetric positive-definite structure.
  \item We evaluate our implementation through numerical experiments across multiple scenario counts and prediction horizons, showing substantial speedups compared with state-of-the-art general sparse direct solvers \texttt{PARDISO}~\cite{pardiso} and \texttt{cuDSS}~\cite{cudss}. 
\end{itemize}

% \begin{figure*}[t]
%     \centering
%      \begin{subfigure}[b]{\columnwidth}
%          \centering
%          \includegraphics[width=\textwidth]{figures/ordering_forward_sequential.pdf}
%          \caption{}
%          \label{fig:ordering_forward_seqential}
%      \end{subfigure}
%      \hfill
%      \begin{subfigure}[b]{\columnwidth}
%          \centering
%          \includegraphics[width=\textwidth]{figures/ordering_reverse_sequential.pdf}
%          \caption{}
%          \label{fig:ordering_reverse_seqential}
%      \end{subfigure}
%      \hfill
%      \begin{subfigure}[b]{\columnwidth}
%          \centering
%          \includegraphics[width=\textwidth]{figures/ordering_reverse_parallel_nosep.pdf}
%          \caption{}
%          \label{fig:ordering_reverse_parallel_nosep}
%      \end{subfigure}
%      \hfill
%      \begin{subfigure}[b]{\columnwidth}
%          \centering
%          \includegraphics[width=\textwidth]{figures/ordering_reverse_parallel_sep.pdf}
%          \caption{}
%          \label{fig:ordering_reverse_parallel_sep}
%      \end{subfigure}
%      \caption{Sparsity of the KKT matrix and its Cholesky factor under different variable orderings. We choose $M=2$ scenarios and $N=9$ for a brief demonstration. All block sizes are assumed to be 3. \FS{the legend is missing}}
%      \label{fig:ordering}
% \end{figure*}

\emph{Notation}
We denote the set of real numbers by \(\mathbb{R}\), and use
\(\mathbb{R}^{n}\) and \(\mathbb{R}^{n\times n}\) for the spaces of
real \(n\)-vectors and real \(n\times n\) matrices, respectively. The
cone of real symmetric positive-definite (SPD)
matrices is denoted by \(\mathbb{S}_{++}^{n}\).
We use \(\mathbb{I}_{m}^{n}\coloneqq\{m,m+1,\ldots,n\}\) for integer index sets. Throughout the paper, \(M\) and \(N\) are reserved for the number of scenarios and their common horizon length, respectively. Accordingly, \(i\in\mathbb{I}_{1}^{M}\) always indexes a scenario, whereas \(k\in\mathbb{I}_{0}^{N}\) always indexes a temporal stage. A scenario superscript, as in \(x_k^i\), is an index rather than an exponent. All vectors are column vectors, and parentheses denote vertical concatenation, e.g., \((x_1,\ldots,x_N)\coloneqq[x_1^\top,\ldots,x_N^\top]^\top\). We use \(I_n\) for the \(n\times n\) identity matrix and \(\operatorname{blkdiag}(\cdot)\) for block-diagonal concatenation.

%%%%%%%%%%%%%%%%%%%%%%%%%%%%%%%%%%%%%%%%%%%%%%%%%%%%%%%%%%%%%%%%%%%%%%%%%%%%%%%%
% ============================== SECTION: Preliminaries ==============================
%%%%%%%%%%%%%%%%%%%%%%%%%%%%%%%%%%%%%%%%%%%%%%%%%%%%%%%%%%%%%%%%%%%%%%%%%%%%%%%%
\section{Preliminaries}

% \subsection{Parallel Direct Solution of Block-Tridiagonal Systems}
\label{sec:parallel_block_tridiagonal}

We briefly introduce the parallel Cholesky factorization and triangular solve for block tridiagonal systems. 
Consider an SPD block-tridiagonal matrix $\Psi$ with $N$ diagonal blocks:
\begin{equation*}
    \underbrace{
    \begin{bmatrix}
    D_{1} & E_{1}^{\top} &                  & \\[-0.2em]
    E_{1} & D_{2}        & \ddots           & \\[-0.2em]
            & \ddots         & \ddots           & E_{N-1}^{\top} \\[0.2em]
            &                & E_{N-1}        & D_{N}
    \end{bmatrix}
    }_{\Psi}
    \underbrace{
    \begin{bmatrix} \xi_1 \\ \xi_2 \\ \vdots \\ \xi_N \end{bmatrix}
    }_{\xi}
    =
    \underbrace{
    \begin{bmatrix} r_1 \\ r_2 \\ \vdots \\ r_N \end{bmatrix}
    }_{r}
\end{equation*}
For any block-wise permutation matrix \(\Gamma\), define
\begin{equation*}
    \widehat \xi \coloneqq \Gamma \xi,\quad
    \widehat r \coloneqq \Gamma r,\quad
    \widehat\Psi \coloneqq \Gamma\Psi\Gamma^\top.
\end{equation*}
The original system is equivalent to \(\widehat\Psi\widehat \xi=\widehat r\) since \(\Gamma^\top\Gamma=I\).
The vector \(\widehat \xi\) contains the same variable blocks in the sequence specified by the rows of \(\Gamma\). Thus, symmetrically permuting \(\Psi\) while applying the same permutation to \(\xi\) and \(r\) is equivalent to reordering the block variables.

A block-wise recursive permutation is proposed by~\cite{schwan2026socu} to expose parallelism in factorization and triangular solves. We denote its \(N\)-block instance by \(\Gamma_N\). At each recursive call, the permutation takes every other block from the current ordered sequence, starting with the first, while preserving their relative order. It then treats the unselected blocks as a new ordered sequence and repeats the same operation:
\begin{equation*}
\Gamma_N(\xi_1,\ldots,\xi_N) = 
\bigl(\xi_1,\xi_3,\xi_5,\ldots,
\Gamma_{\lfloor N/2\rfloor}(\xi_2,\xi_4,\xi_6,\ldots)
\bigr).
\end{equation*}

Take $N=7$ as an example, we have
\[
\Gamma_7 (\xi_1,\ldots,\xi_7)
=
(\xi_1,\xi_3,\xi_5,\xi_7\,|\,\xi_2,\xi_6\,|\,\xi_4),
\]
where the vertical bars show the groups produced by the recursive
calls.  Applying \(\Gamma_7\) symmetrically to \(\Psi\) gives

\begin{equation}
\setlength{\arraycolsep}{3pt}
\Gamma_7 \Psi \Gamma_7^\top =
\left[
\begin{array}{cccc|cc|c}
D_1 &     &     &     & E_1^\top &            &            \\
    & D_3 &     &     & E_2      &            & E_3^\top   \\
    &     & D_5 &     &          & E_5^\top   & E_4        \\
    &     &     & D_7 &          & E_6        &            \\
\hline
E_1 & E_2^\top & &    & D_2      &            &            \\
    &     & E_5 & E_6^\top &    & D_6        &            \\
\hline
    & E_3 & E_4^\top & &        &            & D_4
\end{array}
\right].
\label{eq:seven_block_permutation}
\end{equation}

The groups produced by this recursion induce multiple levels of parallelism used by the factorization and triangular solves.  In this example, the blocks \(1,3,5,7\) are mutually uncoupled and can be processed
concurrently.  After their contributions have been applied, blocks \(2\) and \(6\) can be processed concurrently, followed by block \(4\).  Because each recursive step reduces the number of remaining
blocks by approximately one half, the permutation enables \(O(\log N)\) time complexity with sufficient parallel resources.
We refer interested readers to~\cite{schwan2026socu} for more details.

%%%%%%%%%%%%%%%%%%%%%%%%%%%%%%%%%%%%%%%%%%%%%%%%%%%%%%%%%%%%%%%%%%%%%%%%%%%%%%%%
% ============================== SECTION: Root-Coupled Block-Tridiagonal Systems ==============================
%%%%%%%%%%%%%%%%%%%%%%%%%%%%%%%%%%%%%%%%%%%%%%%%%%%%%%%%%%%%%%%%%%%%%%%%%%%%%%%%
\section{Problem Structure}
\label{sec:system_structure}

\subsection{Problem Formulation}
We consider the problem formulation
\begin{equation}
\label{eq:scenario_problem}
\begin{aligned}
\min_{x_0,\{x^i_{1:N}\}_{i=1}^M} \quad
& \sum_{i=1}^M J_i(x_0,x^i_{1:N})  \\
\text{s.t.} \qquad ~~
& f_k^i(x_k^i,x_{k+1}^i)=0,
&& \forall k\in\mathbb{I}_1^{N-1}, ~i\in\mathbb{I}_1^M, \\
& h_k^i(x_k^i,x_{k+1}^i)\leq 0,
&& \forall k\in\mathbb{I}_1^{N-1}, ~i\in\mathbb{I}_1^M, \\
& f_0^i(x_0, x_1^i)=0,
&& \forall i\in\mathbb{I}_1^M, \\
& h_0^i(x_0, x_1^i)\leq 0,
&& \forall i\in\mathbb{I}_1^M, \\
& f_N^i(x_N^i)=0,
&& \forall i\in\mathbb{I}_1^M, \\
& h_N^i(x_N^i)\leq 0,
&& \forall i\in\mathbb{I}_1^M,
\end{aligned}
\end{equation}

where \(x_0\) is the shared root block\footnote{The initial condition constraint that arises in a scenario-based optimal control problem can be described by $f_0^i(x_0, x_1^i)$.}, while \(x^i_{1:N} \coloneqq (x^i_1, \ldots, x^i_N)\) denotes the scenario-dependent tail variables as illustrated in Fig.~\ref{fig:branch_mpc}.
The cost for the $i$-th scenario is
\begin{equation*}
J_i(x_0,x^i_{1:N})
\coloneqq
\ell_0^i(x_0, x_1^i) +
\sum_{k=1}^{N-1} \ell_k^i(x_k^i,x_{k+1}^i) + \ell_N^i(x_N^i).
\end{equation*}

% \begin{equation}
%     \label{eq:scenario_problem}
%     \begin{aligned}
%         \min_{x^{1:M}_{0:N}} \quad &  \sum_{i=1}^M J_i(x^i_{0:N})  \\
%         \text{s.t.} \quad & f_i(x_k^i, x_{k+1}^i) = 0, & \forall k\in\mathbb{I}_0^{N-1}, i\in\mathbb{I}_1^M, \\
%         & h_i(x_k^i, x_{k+1}^i) \leq 0, & \forall k\in\mathbb{I}_0^{N-1}, i\in\mathbb{I}_1^M,  \\
%         & f_N(x_N^i) = 0, & \forall i\in\mathbb{I}_1^M, \\
%         & h_N(x_N^i) \leq 0, & \forall i\in\mathbb{I}_1^M, \\
%         & x_0^i = \bar x_0, & \forall i\in\mathbb{I}_1^M,
%     \end{aligned}
% \end{equation}
% where the cost for the $i$-th scenario is the sum of stage-wise costs $\ell_k^i$ for all stages:
% \begin{equation}
%     J_i(x^i_{0:N}) \coloneqq \sum_{k=0}^{N-1} \ell_k^i(x_k^i, x_{k+1}^i) + \ell_N^i(x_N^i),
% \end{equation}

Notice that $x_k^i$ contains both states and inputs of the $i$-th scenario at the $k$-th stage\footnote{Except for the last stage, where the inputs do not appear.} for a compact formulation. 
% The formulation~\eqref{eq:scenario_problem} is a general description of the problems arising in scenario-based MPC~\FS{cite some papers}, since it can describe the neighbor-stage couplings within each scenario not only through the dynamics, but also through constraints and costs. 

\begin{assumption}
\label{ass:equal_size}
    We assume that the shared root and all tail variables have unified size, i.e.,
    $x_0\in\mathbb{R}^{n_b}$ and $x_k^i \in\mathbb{R}^{n_b},~ \forall i\in\mathbb{I}_1^M, k\in\mathbb{I}_1^N$. 
\end{assumption}
\begin{remark}
    Although the last stage $x_N^i$ is smaller since it has no input, it can be padded to satisfy Assumption \ref{ass:equal_size}.
\end{remark}

\subsection{Structure of the Linear System}
The optimization methods considered in this work require repeated solutions of linear systems
\begin{equation}
K\Delta x=r,
\label{eq:linear-system}
\end{equation}
where $K$ is SPD (explained in Appendix~\ref{app:equivalence_lin_sys_sparsity}) and $\Delta x$ is the unknown to be solved. 
For each scenario \(i\), we define the stacked tail:
\[
\Delta x^i:=(\Delta x_1^i,\ldots,\Delta x_N^i)
\in\mathbb R^{Nn_b}.
\]
Grouping the tails and placing the root last gives
\[
\Delta x:=(\Delta x^1,\ldots,\Delta x^M,\Delta x_0) \in\mathbb R^{(MN+1)n_b}.
\]
Similarly, we partition the right-hand side as
\[
r=(r^1,\ldots,r^M,r_0),\quad
r^i=(r_1^i,\ldots,r_N^i),
\]
where \(r_k^i,r_0\in\mathbb R^{n_b}\).
% We group the variables associated with each tail $\Delta x^i_{1:N}$ and place the root variables $\Delta x_0$ last:
% \begin{equation}
% \Delta x \coloneqq
% (\Delta x^1_{1:N}, \Delta x^2_{1:N}, \ldots, 
% \Delta x^M_{1:N}, \Delta x_0
% )
% \in \mathbb{R}^{(MN+1)n_b}.
% \label{eq:variable-partition}
% \end{equation}
With this ordering, the matrix becomes
\begin{equation}
K=
\begin{bmatrix}
K_1 &     &        &     & C_1^{\top} \\
    & K_2 &        &     & C_2^{\top} \\
    &     & \ddots &     & \vdots \\
    &     &        & K_M & C_M^{\top} \\
C_1 & C_2 & \cdots & C_M & R
\end{bmatrix}
\in\mathbb{S}_{++}^{(MN+1)n_b}.
\label{eq:root_coupled_matrix}
\end{equation}
Here, $K_i$ is the temporal system associated with tail $i$, $R\in \mathbb{S}_{++}^{n_b}$ is the root block associated with the variables $\Delta x_0$, and $C_i$ couples tail $i$ to the root. Different tails have no direct coupling and interact only through the root block. 
Each tail matrix $K_i$ has the block-tridiagonal structure:
\begin{equation}
K_i = 
\begin{bmatrix}
D_{i,1} & E_{i,1}^{\top} &                  & \\
E_{i,1} & D_{i,2}        & \ddots           & \\
        & \ddots         & \ddots           & E_{i,N-1}^{\top} \\[0.5em]
        &                & E_{i,N-1}        & D_{i,N}
\end{bmatrix}
\in\mathbb{S}_{++}^{Nn_b}.
\label{eq:tail_block_tridiagonal}
\end{equation}
% The block-tridiagonal pattern follows from the local coupling between consecutive stages in a multiple-shooting optimal-control formulation. A stage without a physical input can be padded so that every temporal block has the same dimension without changing the unpadded solution.

The shared root variables couple directly to only the first temporal block of each tail, making the matrix $C_i$ contain only one non-zero block $C_{i,1}\in\R^{n_b\times n_b}$ at the beginning:
\begin{equation*}
C_i=\begin{bmatrix} C_{i,1}&0&\cdots&0\end{bmatrix}
\in\mathbb{R}^{n_b \times Nn_b}.
\end{equation*}

The resulting coefficient matrix has a block-diagonal-arrow structure with block-tridiagonal tails and a single root-coupling block in each tail. The next section develops a direct solver tailored to this structure.

%%%%%%%%%%%%%%%%%%%%%%%%%%%%%%%%%%%%%%%%%%%%%%%%%%%%%%%%%%%%%%%%%%%%%%%%%%%%%%%%
% ============================== SECTION: Two-Level Parallel Direct Solution ==============================
%%%%%%%%%%%%%%%%%%%%%%%%%%%%%%%%%%%%%%%%%%%%%%%%%%%%%%%%%%%%%%%%%%%%%%%%%%%%%%%%
\section{Two-Level Parallel Direct Solution}
\label{sec:parallel_direct_solution}

\begin{figure*}[t]
\centering
\includegraphics[width=0.85\linewidth]{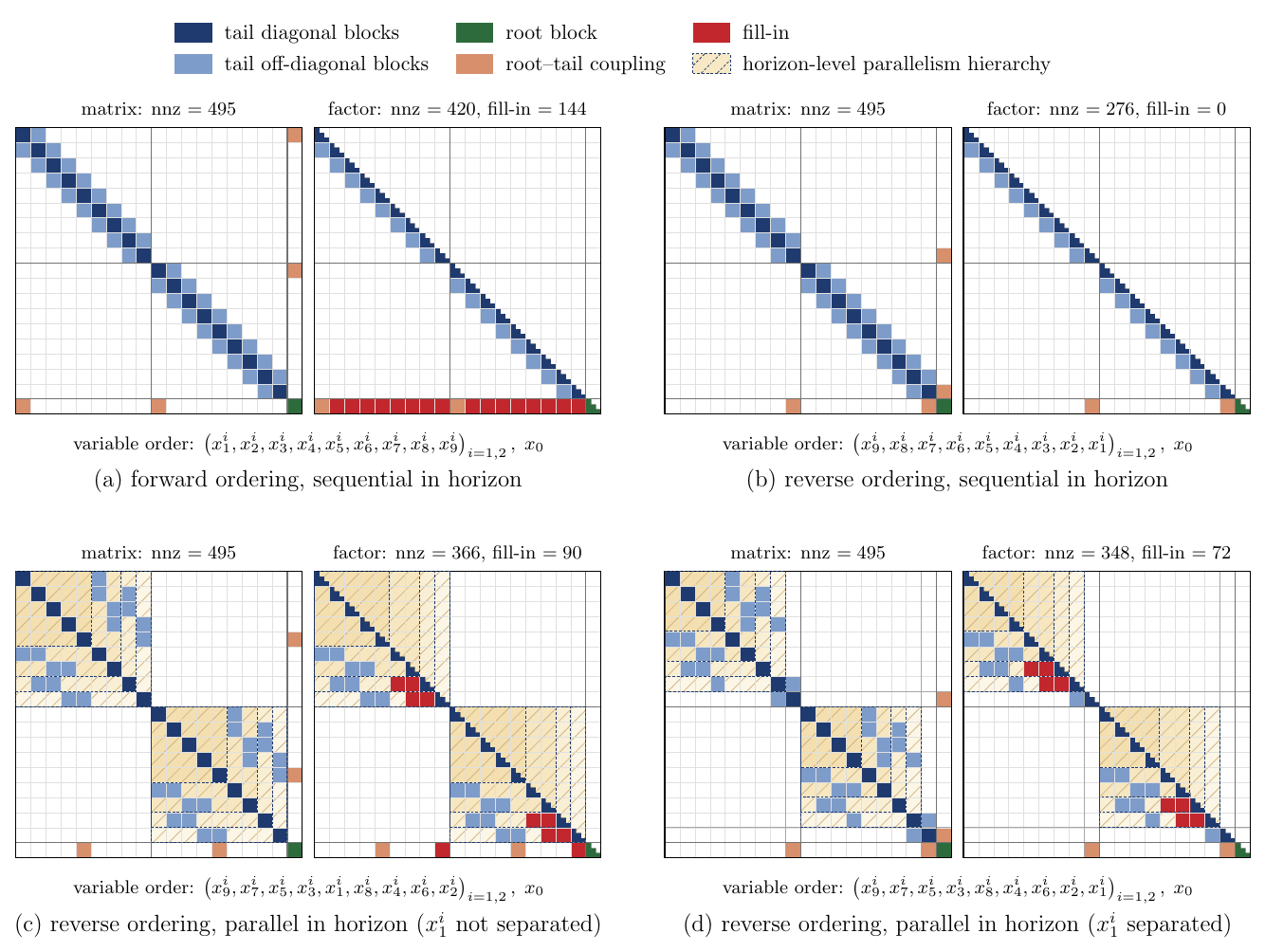}
\caption{Sparsity of the matrix in the linear system to be solved and its Cholesky factor under different variable orderings. We choose $M=2$ scenarios and $N=9$ for a brief demonstration. All block sizes are assumed to be 3.}
\label{fig:ordering}
\end{figure*}

The root-coupled system in~\eqref{eq:root_coupled_matrix} exposes two distinct sources of parallelism:
\begin{itemize}
    \item \emph{Scenario-level parallelism (SLP)}: the tail matrices \(K_i\) have no direct coupling with each other and can be processed concurrently. 
    \item \emph{Horizon-level parallelism (HLP)}: the natural temporal variable order creates a block-tridiagonal structure within each \(K_i\), whose Cholesky factorization and triangular solve can be done in parallel as studied by~\cite{schwan2026socu}.
\end{itemize}
This section introduces a block permutation that exposes parallelism within every tail without spreading its coupling to the root variables. The resulting factorization and triangular solves exploit both levels of parallelism and synchronize the tails only through operations on the root block.

\subsection{Block Permutation for Two-Level Parallelism}
\label{sec:block_permutation_two_level}

To expose parallelism, we apply a dedicated permutation to the matrix \(K\). Let \(\Pi\) denote the permutation matrix of the complete system. The original system is then equivalent to
\begin{equation}
    \underbrace{\Pi K\Pi^\top}_{\widehat K}
    \underbrace{\Pi\Delta x}_{\Delta\widehat x}
    =
    \underbrace{\Pi r}_{\widehat r},
    \label{eq:permuted_linear_system}
\end{equation}
where $\widehat K \coloneqq \Pi K \Pi^\top$ is the permuted matrix, $\Delta \widehat x \coloneqq \Pi \Delta x$ is the permuted unknown and $\widehat r \coloneqq \Pi r$ are the permuted right-hand-side (RHS).
The rows of \(\Pi\) specify the sequence in which the block variables
appear in \(\Delta \widehat x\) and $\widehat r$.  Consequently, factorizing \(\widehat K\) is
algebraically equivalent to factorizing \(K\) according to that
variable order.

We first reverse the order of variables corresponding to tail \(i\) on the block-level, and the reason for doing this will become clear later: 
\begin{equation*}
    \Omega_i \Delta x^i =
    \left(\Delta x_N^i, \ldots, \Delta x_1^i\right), \quad
    \Omega_i r^i =
    \left(r_N^i, \ldots,  r_1^i\right),
\end{equation*}
where \(\Omega_i\in\mathbb{R}^{Nn_b\times Nn_b}\) is the corresponding block-reversal permutation matrix. 
Then we construct the block-wise permutation $\Gamma_{N-1}$ from~\cite{schwan2026socu} defined in Section~\ref{sec:parallel_block_tridiagonal} and apply it onto $(\Delta x_N^i, \ldots, \Delta x_2^i)$ and $(r_N^i,\ldots, r_2^i)$. 
The block \( \Delta x_1^i\), which is the only tail block coupled to the root variables $\Delta x_0$, is excluded from \(\Gamma_{N-1}\) and kept last. 

The permutation matrix for the $i$-th tail is therefore
\begin{equation*}
    \Pi_i =
    \begin{bmatrix}
        \Gamma_{N-1} & 0 \\
         0  & I_{n_b}
    \end{bmatrix}\Omega_i,
\end{equation*}
and we denote its corresponding permuted unknown and RHS associated with the $i$-th tail by
\[
\Delta\widehat x^i:=\Pi_i\Delta x^i,
\quad
\widehat r^i := \Pi_i r^i.
\]
Since we consider the case where all tails have the same horizon length and block dimensions, the
same permutation pattern is applied independently to every tail. The complete permutation matrix is therefore
\begin{equation*}
    \Pi= \operatorname{blkdiag} \left(\Pi_1,\Pi_2,\ldots,\Pi_M,I_{n_b}\right),
\end{equation*}
and the permuted matrix is given by
\begin{equation}\label{eq:permuted_root_coupled_matrix}
\underbrace{
\begin{bmatrix}
\widehat K_1 &               &        &              & \widehat C_1^\top \\
             & \widehat K_2  &        &              & \widehat C_2^\top \\
             &               & \ddots &              & \vdots \\
             &               &        & \widehat K_M & \widehat C_M^\top \\
\widehat C_1 & \widehat C_2  & \cdots & \widehat C_M & R
\end{bmatrix}
}_{\widehat K }
\underbrace{
\begin{bmatrix}
    \Delta \widehat x^1 \\
    \Delta \widehat x^2 \\
    \vdots \\
    \Delta \widehat x^M \\
    \Delta x_0
\end{bmatrix}
}_{\Delta \widehat x}
=
\underbrace{
\begin{bmatrix}
    \widehat r^1 \\
    \widehat r^2 \\
    \vdots \\
    \widehat r^M \\
    r_0
\end{bmatrix}
}_{\widehat r},
\end{equation}

% \begin{equation}
% \widehat K \coloneqq \Pi K \Pi^\top =
% \begin{bmatrix}
% \widehat K_1 &               &        &              & \widehat C_1^\top \\
%              & \widehat K_2  &        &              & \widehat C_2^\top \\
%              &               & \ddots &              & \vdots \\
%              &               &        & \widehat K_M & \widehat C_M^\top \\
% \widehat C_1 & \widehat C_2  & \cdots & \widehat C_M & R
% \end{bmatrix},
% \label{eq:permuted_root_coupled_matrix}
% \end{equation}
which retains the block-diagonal-arrow structure in $\widehat K$ and preserves scenario-level independence.
Each tail is given by \(\widehat K_i=\Pi_iK_i\Pi_i^\top\), which allows exposing horizon-level
parallelism independently inside all tails as described in Section~\ref{sec:parallel_block_tridiagonal}.

Moreover, since \(\Delta x_1^i\) remains last in every tail,
the permuted coupling matrix $\widehat C_i$ has the following form:
\begin{equation}
    \widehat C_i=C_i\Pi_i^\top =
    \begin{bmatrix}
        0 & \cdots & 0 & C_{i,1}
    \end{bmatrix}
    \in\mathbb{R}^{n_b \times Nn_b},
    \label{eq:permuted_tail_coupling}
\end{equation}
which contains only one non-zero block in the end.

Figure~\ref{fig:ordering} illustrates the effect of variable reversal and separated treatment of \(\Delta x_1^i\).
Under forward temporal ordering, eliminating \(\Delta x_1^i\) creates a coupling between the root and \(\Delta x_2^i\), causing subsequent eliminations to propagate this coupling along the tail as shown in Figure~\ref{fig:ordering}(a), creating considerable fill-in blocks.
Figure~\ref{fig:ordering}(b) reverses the variables in each tail to avoid the fill-in, but the factorization remains sequential along the horizon.  Applying the recursive block permutation to the complete reversed tail as in Figure~\ref{fig:ordering}(c) exposes horizon-level parallelism, but allows \(\Delta x_1^i\) to enter an earlier elimination level and again spreads the root coupling.  The proposed permutation, as illustrated in Figure~\ref{fig:ordering}(d), applies \(\Gamma_{N-1}\) only to \(\Delta x_N^i,\ldots,\Delta x_2^i\) and fixes \(\Delta x_1^i\) at the end.
Placing \(\Delta x_1^i\) last ensures that the root–tail coupling remains confined to a single block in the Cholesky factor, as established in Section IV-B, while the remaining tail variables admit horizon-parallel elimination.

\subsection{Parallel Cholesky Factorization}

\begin{algorithm}[!t]
\caption{Two-Level Parallel Cholesky Factorization}
\label{alg:parallel_factorization}
\begin{algorithmic}[1]
\Require Tail matrices $\{K_i\}_{i=1}^{M}$, coupling blocks $\{C_{i,1}\}_{i=1}^{M}$, root block $R$, and permutations $\{\Pi_i\}_{i=1}^{M}$
\ForAll{$i=1,\ldots,M$ in parallel} \Comment{SLP}
    \State Compute the permuted matrix $\widehat K_i$
    \State $L_i \gets \operatorname{chol}(\widehat K_i)$ \Comment{Tail factor., $\mathcal{O}(n_b^3 \log N)$, HLP}
    \State $F_{i,1} \gets C_{i,1} L_{i,1}^{-\top}$ \Comment{Tail coupling solve, $\mathcal{O}(n_b^3)$}
    \State $U_i\gets F_{i,1}F_{i,1}^{\top}$ \Comment{Local Schur update, $\mathcal{O}(n_b^3)$}
\EndFor
\State $\widetilde R\gets R-\sum_{i=1}^{M}U_i$ \Comment{Root reduction, $\mathcal{O}(n_b^2 \log M)$}
\State $L_R \gets \operatorname{chol}(\widetilde R)$ \Comment{Root factorization, $\mathcal{O}(n_b^3)$}
\State \Return $\{L_i\}_{i=1}^{M}$, $\{F_{i,1}\}_{i=1}^{M}$, and $L_R$
\end{algorithmic}
\end{algorithm}

We consider the permuted matrix \(\widehat K\)
in~\eqref{eq:permuted_root_coupled_matrix} and compute its Cholesky
factor \(L\) such that \(LL^\top=\widehat K\), where $L$ is lower triangular and inherits the block-diagonal-arrow form:
\begin{equation*}
L=
\begin{bmatrix}
L_1 &     &        &     &  \\
    & L_2 &        &     &  \\
    &     & \ddots &     &  \\
    &     &        & L_M &  \\
F_1 & F_2 & \cdots & F_M & L_R
\end{bmatrix}.
\end{equation*}

Matching the blocks in \(LL^\top=\widehat K\) yields the following five steps:
\begin{enumerate}
    \item Tail factorization: \(L_i\gets \operatorname{chol}(\widehat K_i)\)
    for all \(i\in\mathbb{I}_1^M\).
    \item Tail coupling solve: \(F_i \gets \widehat C_i L_i^{-\top}\) for all
    \(i\in\mathbb{I}_1^M\).
    \item Local Schur update: \(U_i\gets F_iF_i^\top\) for all
    \(i\in\mathbb{I}_1^M\).
    \item Root reduction: \(\widetilde R\gets
    R-\sum_{i=1}^M U_i\).
    \item Root factorization: \(L_R\gets \operatorname{chol}(\widetilde R)\).
\end{enumerate}

% \begin{table}[h]
%     \centering
%     \caption{Steps of the parallel factorization.}
%     \label{tab:factorization_steps}
%     \begin{tabular}{@{}ll@{}}
%         \toprule
%         Step & Operation \\
%         \midrule
%         Tail factorization  & \(L_i \gets \operatorname{chol}(\widehat K_i), ~\forall i\in\mathbb{I}_1^M\) \\
%         Tail coupling solve & \(F_i \gets \widehat C_i L_i^{-\top}, ~\forall i\in\mathbb{I}_1^M\) \\
%         Local Schur update  & \(U_i \gets F_i F_i^\top, ~\forall i\in\mathbb{I}_1^M\) \\
%         Root reduction      & \(\widetilde R \gets R - \sum_{i=1}^M U_i\) \\
%         Root factorization  & \(L_R \gets \operatorname{chol}(\widetilde R)\) \\
%         \bottomrule
%     \end{tabular}
% \end{table}

For each tail, the factorization, coupling solve, and local Schur update are performed in sequence. These computations are independent across tails and can therefore proceed concurrently. The root reduction combines their contributions before the root block is factorized.
Specifically, we demonstrate the first three steps in more detail:

\subsubsection{Tail factorization}
Within each tail, the factorization of \(\widehat K_i\) uses the method in~\cite[Algorithm~4]{schwan2026socu} as we introduced in Section~\ref{sec:parallel_block_tridiagonal}, achieving horizon-level parallelization.

\subsubsection{Tail coupling solve}
Since \(\widehat C_i^\top\) has non-zeros only on its last block row, solving \(F_i \gets \widehat C_i L_i^{-\top}\) produces all zeros until the final block of $F_i$, i.e., 
\begin{equation*}
F_i=\begin{bmatrix}0&\cdots&0&F_{i,1}\end{bmatrix}
\in\mathbb{R}^{n_b \times Nn_b}.
\end{equation*}
% where the last block corresponds to \(\Delta x_1^i\)
Therefore, when computing $F_i \gets \widehat C_i L_i^{-\top}$, instead of using the complete matrices, it can be reduced to: 
\begin{equation*}
    F_{i,1} \gets C_{i,1} L_{i,1}^{-\top},
\end{equation*}
where $L_{i,1}\in\mathbb{R}^{n_b\times n_b}$ is the final diagonal block of $L_i$, associated with the original stage variable $\Delta x_1^i$.
A formal proof is given in Appendix~\ref{app:proof_tail_coupling_solve}.

\subsubsection{Local Schur update}
Since $F_i$ only contains one non-zero block $F_{i,1}$, $U_i \gets F_i F_i^\top$ can be simplified as
\begin{equation*}
U_{i} \gets F_{i,1} F_{i,1}^\top.
\end{equation*}

Algorithm~\ref{alg:parallel_factorization} summarizes the two-level parallel Cholesky factorization, with the time complexity of each step annotated in the comments, assuming sufficient parallel computing units are available. By summing them up, the overall time complexity of the proposed Cholesky factorization is 
$\mathcal{O}(n_b^3 \log N + n_b^2\log M)$.
% $\mathcal{O}(n_b^3 \log N)$.

\begin{algorithm}[!t]
\caption{Two-Level Parallel Triangular Solves}
\label{alg:parallel_solve}
\begin{algorithmic}[1]
\Require Tail factors $\{L_i\}_{i=1}^{M}$, coupling factors \(\{F_{i,1}\}_{i=1}^M\), permutations
\(\ \{\Pi_i\}_{i=1}^M\), root factor \(L_R\) and RHS \(r\)
\Statex \textit{Forward substitution}
\ForAll{\(i=1,\ldots,M\) in parallel} \Comment{SLP}
    \State \(\widehat r^i \gets \Pi_i r^i\)
    \State \(\Delta y^i \gets L_i^{-1} \widehat r^i\) \Comment{Tail fwd subs., $\mathcal{O}(n_b^2 \log N)$, HLP}
\EndFor
\State \(\widetilde r_0 \! \gets \! r_0 \!- \! \sum_{i=1}^{M} \! F_{i,1} \Delta y_{1}^i\) \Comment{Root reduc., $\mathcal{O}(n_b^2 \!+\! n_b \! \log \! M)$}
\State \(\Delta y_0 \gets L_R^{-1} \widetilde r_0\) \Comment{Root fwd subs., $\mathcal{O}(n_b^2)$}
\Statex \textit{Backward substitution}
\State \(\Delta x_0 \gets L_R^{-\top} \Delta y_0\) \Comment{Root bwd subs., $\mathcal{O}(n_b^2)$}
\ForAll{\(i=1,\ldots,M\) in parallel} \Comment{SLP}
    \State \( \Delta y_{1}^i \gets \Delta y_{1}^i - F_{i,1}^\top\Delta x_0\)  \Comment{Root prop., $\mathcal{O}(n_b^2)$}
    \State \( \Delta\widehat x^i \! \gets \! L_i^{-\top} \!\! \Delta y^i\)   \Comment{Tail bwd subs., $\mathcal{O}(n_b^2 \log \! N)$, HLP}
    \State \(\Delta x^i \gets \Pi_i^\top\Delta\widehat x^i\)
\EndFor
\State \(\Delta x\gets (\Delta x^1,\ldots,\Delta x^M,\Delta x_0)\)  \Comment{Assemble solution}
\State \Return \(\Delta x\)
\end{algorithmic}
\end{algorithm}

\subsection{Parallel Triangular Solve}
Given the permuted right-hand side \(\widehat r\), we solve
\[
L\Delta y=\widehat r,
\quad
L^\top\Delta\widehat x=\Delta y,
\]
by forward and backward substitution, respectively. Here,
\[
\Delta y=(\Delta y^1,\ldots,\Delta y^M,\Delta y_0),
\]
where \(\Delta y^i\in\mathbb R^{Nn_b}\) and \(\Delta y_0\in\mathbb R^{n_b}\). 
The solution consists of six stages:
\begin{enumerate}
    \item Tail forward substitution:
    \(\Delta y^i \gets L_i^{-1} \widehat r^i\) for all \(i\in\mathbb{I}_1^M\).
    \item Root reduction:
    % \(\widetilde r_0\gets r_0-\sum_{i=1}^M F_{i,1} \Delta y_{1}^i\),
    % where \(\Delta y_{1}^i\) is the last block of \(\Delta y_i\).
    \(\widetilde r_0\gets r_0-\sum_{i=1}^M F_{i} \Delta y^i\).
    \item Root forward substitution: 
    \(\Delta y_0 \gets L_R^{-1} \widetilde r_0\).
    \item Root backward substitution:
    \(\Delta x_0 \gets L_R^{-\top} \Delta y_0\).
    \item Root propagation:
    % \( \Delta \tilde{y}_{1}^i \gets \Delta y_{1}^i - F_{i,1}^\top\Delta x_0\) for all \(i\in\mathbb{I}_1^M\).
    \( \Delta \tilde{y}^i \gets \! \Delta y^i -\! F_{i}^\top\Delta x_0\) for all \(i\in\mathbb{I}_1^M\).
    \item Tail backward substitution:
    \(\Delta\widehat x_i \gets L_i^{-\top} \Delta \tilde{y}_i\), for all \(i\in\mathbb{I}_1^M\).
    % \item Inverse permutation:
    % \(\Delta x_i\gets\Pi_i^\top\Delta\widehat x_i\) for all
    % \(i\in\mathbb{I}_1^M\), with
    % \(\Delta x_0\gets\Delta\widehat x_0\).
\end{enumerate}

The tail operations in the first, fifth, and sixth stages are independent across tails and can therefore be performed concurrently at the scenario level. Within each tail, the forward and backward solves use the horizon-level parallel algorithm in~\cite[Algorithm~5]{schwan2026socu}.

Similar to the factorization, since \(F_i=[\,0\ \cdots\ 0\ F_{i,1}\,]\), the root reduction depends only on the last block of \(\Delta y^i\), and root propagation modifies only that block, we can rewrite the root reduction step as
\(
    \widetilde r_0\gets r_0-\sum_{i=1}^M F_{i,1} \Delta y_{1}^i
\)
and the root propagation step as
\(
  \Delta \tilde{y}_{1}^i \gets \Delta y_{1}^i - F_{i,1}^\top\Delta x_0  ~ \text{for all } i\in\mathbb{I}_1^M,
\)
where \(\Delta y_{1}^i\) is the last block of \(\Delta y_i\). 
% The subscript \(1\) in \(\Delta y_1^i\) refers to the original temporal stage, whose block occupies the last position after permutation.

% also limits the interaction between each
% tail and the root during repeated solves. The forward coupling product
% is \(F_{i,1}y_{i,1}\), and the backward coupling correction
% \(F_{i,1}^\top\Delta\widehat x_0\) is applied only to the last block of
% the tail right-hand side. Thus, only one block per tail is exchanged
% with the root; the subsequent solves with the complete \(L_i\) still
% exploit horizon-level parallelism.

Algorithm~\ref{alg:parallel_solve} summarizes the two-level parallel triangular solve, and the time complexity per step is marked in the comments, assuming sufficient parallel computing units. The overall time complexity is $\mathcal{O}(n_b^2 \log N + n_b \log M)$.

%%%%%%%%%%%%%%%%%%%%%%%%%%%%%%%%%%%%%%%%%%%%%%%%%%%%%%%%%%%%%%%%%%%%%%%%%%%%%%%%
% ============================== SECTION: Numerical Experiments ==============================
%%%%%%%%%%%%%%%%%%%%%%%%%%%%%%%%%%%%%%%%%%%%%%%%%%%%%%%%%%%%%%%%%%%%%%%%%%%%%%%%
\section{Numerical Experiments}
\label{sec:experiments}

\begin{figure*}[!t]
    \centering
     \begin{subfigure}[b]{0.95\columnwidth}
         \centering
         \includegraphics[width=\textwidth]{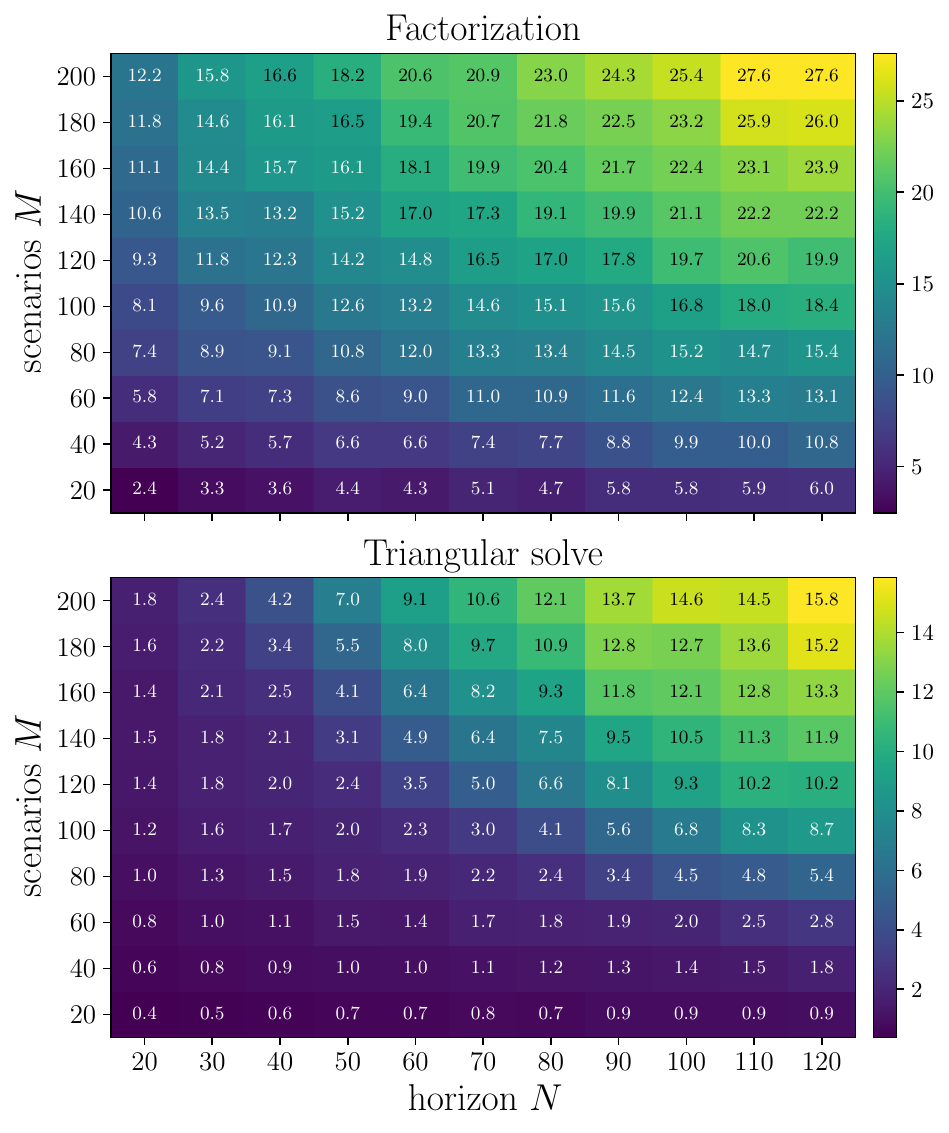}
         \caption{Speedup against \texttt{PARDISO} (8 threads).}
         \label{fig:speedup_pardiso}
     \end{subfigure}
     \hfill
     \centering
     \begin{subfigure}[b]{0.95\columnwidth}
         \centering
         \includegraphics[width=\textwidth]{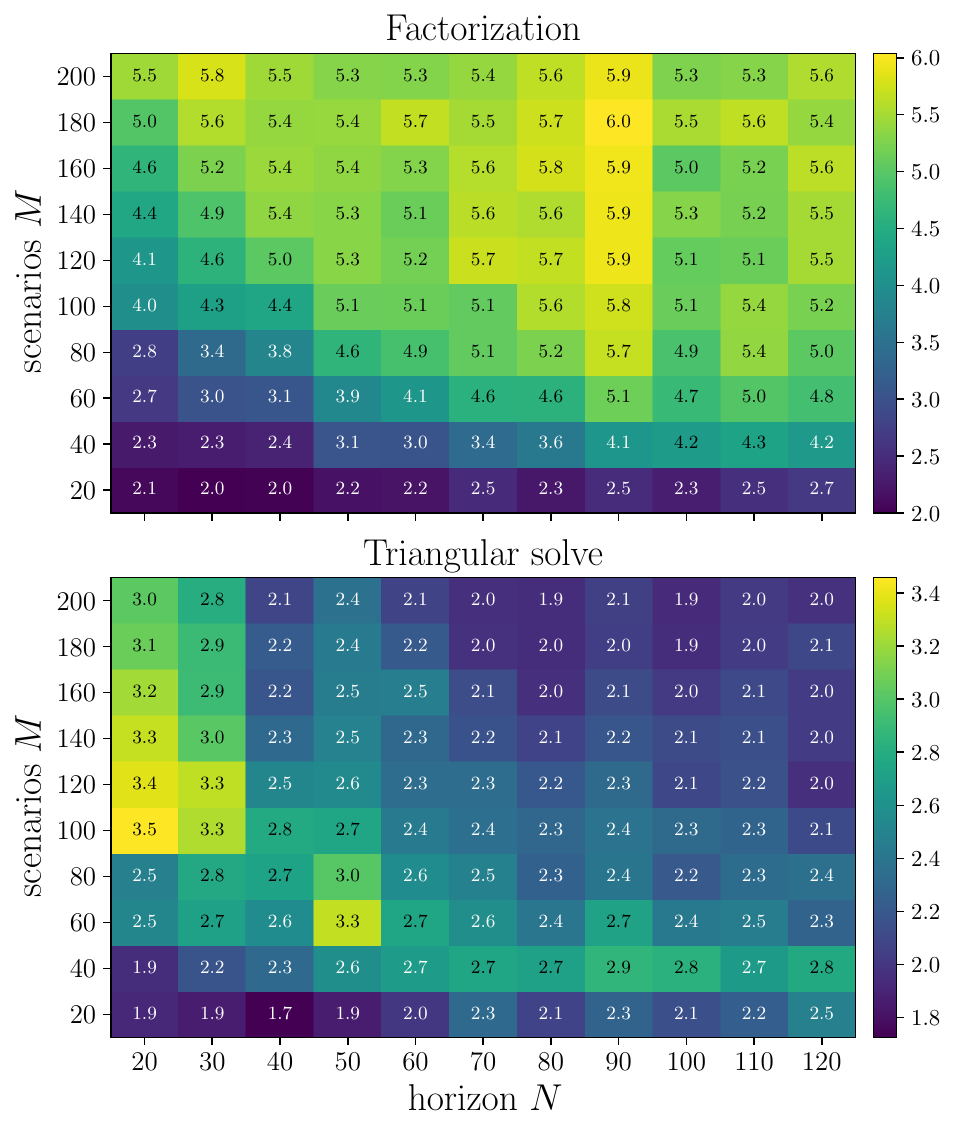}
         \caption{Speedup against \texttt{cuDSS}.}
         \label{fig:speedup_cudss}
     \end{subfigure}
     \caption{Speedups of the proposed parallel factorization and triangular solve compared to \texttt{PARDISO} and \texttt{cuDSS}.}
     \label{fig:speedup_factor_and_solve}
\end{figure*}

To verify the effectiveness of our proposed method, we implement it using NVIDIA Warp~\cite{macklin2022warp}, a Python framework to generate high-performance kernels for NVIDIA GPUs. 
For the horizon-level parallel factorization and triangular solve, we use the existing open-source implementation \texttt{socu}~\cite{schwan2026socu}. 
Fixed sparsity permits a one-time analysis phase that allocates persistent workspaces and generates CUDA kernels. The factorization and solve pipelines are captured separately as CUDA graphs to reduce kernel launch overhead.

All experiments are carried out on a workstation equipped with
% 24-core (8 performance cores and 16 efficiency cores) 
an Intel Core Ultra 9 285K CPU, 64 GB of system memory, and an NVIDIA GeForce RTX 5090 GPU with 32 GB of device memory. All solvers use double-precision (float64) arithmetic.

\subsection{Experiment Settings}

% \begin{figure}[htbp]
%     \centering
%     \includegraphics[width=1.0\linewidth]{figures/quadrotor_endpoint_times_T30_combined.pdf}
%     \caption{Time used for factorization and triangular solve with fixed horizon length $N=30$.}
%     \label{fig:abs_time_fixed_horizon}
% \end{figure}
\begin{figure*}[!t]
    \centering
    \begin{subfigure}[b]{\columnwidth}
        \centering
        \includegraphics[width=\textwidth]{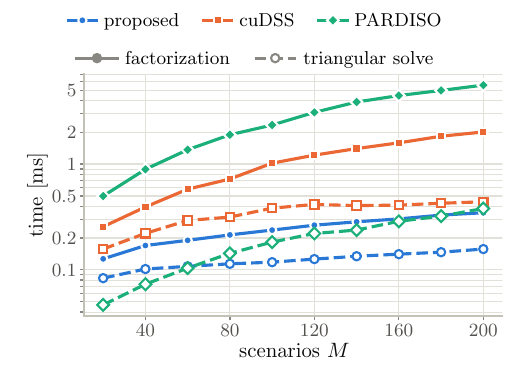}
        \caption{Fixed horizon length $N=30$ with varying $M$.}
        \label{fig:abs_time_fixed_horizon}
    \end{subfigure}
    \hfill
    \begin{subfigure}[b]{\columnwidth}
        \centering
        \includegraphics[width=\textwidth]{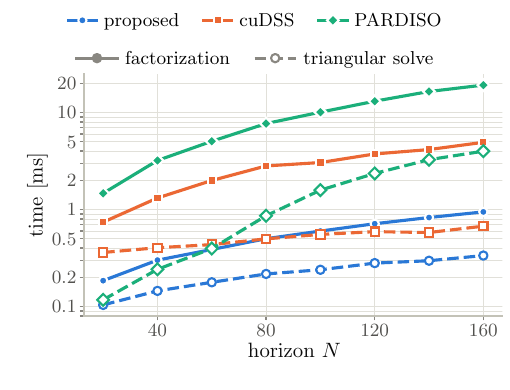}
        \caption{Fixed number of scenarios $M=100$ with varying $N$.}
        \label{fig:abs_time_fixed_scenarios}
    \end{subfigure}
    \caption{Absolute factorization and triangular-solve times along representative slices of the parameter grid.}
    \label{fig:absolute_times}
\end{figure*}

We consider a quadrotor navigating in 3-D space among spherical
obstacles under uncertainty. Its model has 12 states, including
position, velocity, ZYX Euler angles, and body rates, and 4 inputs
corresponding to the rotor thrusts, leading to a block size \(n_b=16\). 
% The continuous-time dynamics include gravity, linear translational and rotational drag, and
% the gyroscopic term, and are discretized using an explicit fourth-order Runge--Kutta method.
% The uncertainty follows the immediate-branching structure introduced
% in Section~\ref{sec:system_structure}. All scenarios share the complete first node
% and evolve independently thereafter. 
Each scenario independently samples a vehicle mass with a 10\% relative standard deviation and a time-invariant wind acceleration.

We sweep $M \in \{20,40,\ldots,200\}$ scenarios and
$N \in \{20,30,\ldots,160\}$ stages to cover a wide range of
configurations. We compare our solver against two state-of-the-art
general sparse direct solvers: Intel MKL \texttt{PARDISO}~\cite{pardiso}%
\footnote{We call \texttt{PARDISO}'s low-level Python binding and avoid dynamic memory allocation. We set \texttt{MKL\_DYNAMIC=FALSE}, \texttt{OMP\_DYNAMIC=FALSE}, and
\texttt{OMP\_WAIT\_POLICY=ACTIVE} to reduce thread scheduling overhead.},
and NVIDIA \texttt{cuDSS}~\cite{cudss}. 
For \texttt{PARDISO}, we use 8 threads bound to the CPU's 8 performance cores, with default nested dissection ordering and iterative refinement turned off.
The one-time costs, including symbolic analysis, kernel generation, memory allocation, and CUDA graph capture, are excluded from the reported time. For our solver and \texttt{cuDSS}, data remain device-resident.
Each value is the median of 50 executions after 10 warm-up runs.

\subsection{Results}
Across all experiments, the infinity norm of the difference between the solutions returned from our implementation and those from the two baseline solvers is around or below $10^{-13}$, demonstrating the correctness within double-precision tolerance. 
The speedup maps in Fig.~\ref{fig:speedup_factor_and_solve} report the speedups relative to the 2 baseline solvers. Relative to \texttt{PARDISO}, the proposed factorization is \(2.4\)--\(27.6\times\) faster over the grid. The triangular-solve speedup ranges from approximately \(0.4\times\) at the smallest problem to \(15.8\times\). When the problem size is small, \texttt{PARDISO} is faster in the triangular solve than the proposed method, which is expected since the workload is not big enough to amortize the fixed overhead of the GPU, such as kernel launch and synchronization. Once either the scenario count or the horizon length increases, the additional scenario- and horizon-level parallel work is sufficient to amortize these costs, and the proposed triangular solve takes advantage of this.

The comparison with \texttt{cuDSS} isolates the benefit over a general-purpose GPU sparse direct solver. The proposed method achieves \(2.0\)--\(6.0\times\) factorization speedups and \(1.7\)--\(3.5\times\) triangular solve speedups, remaining faster at every tested point. The advantage generally increases with \(M\) and \(N\), but less regularly than the speedup pattern in Fig.~\ref{fig:speedup_pardiso}, possibly related to the
superpanel strategy of \texttt{cuDSS}, which groups some columns into dense panels.
The speedups over \texttt{cuDSS} likely arise from exploiting the specific structure of the matrix. Operating on dense blocks may improve memory locality and reduce sparse-indexing overhead, while explicitly scheduling independent tails and parallel elimination levels exposes concurrency at both the scenario and horizon levels. The tailored ordering also confines each root–tail coupling to a single block in the factor, limiting the associated fill-in and data movement. Although \texttt{cuDSS} exploits sparsity through general-purpose analysis and factorization, its ordering and scheduling are not necessarily tailored to these structural properties.

Figure~\ref{fig:abs_time_fixed_horizon} shows the sensitivity of execution time to the scenario count and prediction horizon. Along scenario count, as shown in Figure~\ref{fig:abs_time_fixed_horizon} with fixed \(N=30\), increasing \(M\) from 20 to 200 increases factorization time only from approximately \(0.15\) to \(0.35\) ms and triangular solve time from approximately \(0.10\) to \(0.16\) ms. This weak dependence on scenario count demonstrates the effectiveness of concurrent tail processing on the GPU over the tested range. Along horizon length, as shown in Figure~\ref{fig:abs_time_fixed_scenarios} with fixed \(M=100\), increasing \(N\) from 40 to 160 raises the factorization time and the triangular solve time by \(3.1\times\) and \(2.3\times\) for a four times longer horizon. Over the same range, \texttt{PARDISO}'s times grow by \(6.0\times\) and \(16.6\times\). The factorization remains below \(1\) ms and the triangular solve below \(0.35\) ms at \(M=100\), \(N=160\), highlighting the ability to handle large scenario collections and long
horizons at low absolute latency.

%%%%%%%%%%%%%%%%%%%%%%%%%%%%%%%%%%%%%%%%%%%%%%%%%%%%%%%%%%%%%%%%%%%%%%%%%%%%%%%%
% ============================== SECTION: Conclusion and Future Work ==============================
%%%%%%%%%%%%%%%%%%%%%%%%%%%%%%%%%%%%%%%%%%%%%%%%%%%%%%%%%%%%%%%%%%%%%%%%%%%%%%%%
\section{Conclusion and Future Work}
\label{sec:conclusion}

This paper presented a direct linear-system solver for branch MPC where multiple scenarios are coupled through a single root node. By combining scenario-level and horizon-level parallelism, the proposed method directly accelerates both Cholesky factorization and triangular solve for the underlying linear system in the optimization algorithm on the GPU. The experimental results show substantial speedups over state-of-the-art general sparse direct solvers, highlighting the benefits of our method.

Future work will consider formulations in which the entire input sequence is shared across all scenarios. In that case, the scenario tails are coupled through horizon-wide input variables rather than a single root node, presenting a different sparsity pattern which requires new methods to exploit the potential parallelism.

\bibliographystyle{IEEEtran}
\bibliography{IEEEabrv,references}

\appendices

\section{Linear System Sparsity in IPM and ADMM}
\label{app:equivalence_lin_sys_sparsity}
% \begin{equation}
% \begin{aligned}
%     \min_x \quad & \ell(x) \\
%     \text{s.t.} \quad & f(x) = 0 \\
%                     & h(x) \leq 0
% \end{aligned}
% \end{equation}
% For the interior-point method, the KKT system is:
% \begin{equation}
%     \begin{bmatrix}
%         \nabla_x^2 \ell + \delta_pI & (\nabla f)^\top & (\nabla h)^\top \\
%         \nabla_x f & -\delta_d I & \\
%         \nabla_x h &  &-\delta_d I
%     \end{bmatrix}
% \end{equation}

For nonlinear instances of~\eqref{eq:scenario_problem}, sequential quadratic programming (SQP) generates a sequence of QP subproblems by linearizing the constraints and constructing a quadratic approximation of the objective. Therefore, we consider QP subproblems that preserve the coupling structure of~\eqref{eq:scenario_problem}. It therefore suffices to study the linear systems arising from the following QP:
\begin{equation}
\label{eq:appendix_qp}
\begin{aligned}
\min_x\quad & \tfrac{1}{2}x^\top P x+c^\top x\\
\text{s.t.}\quad & Ax=b, \quad  h_l \leq Gx\leq h_u,
\end{aligned}
\end{equation}
where $P$ is symmetric positive semi-definite, $A$ is the equality-constraint matrix, and $G$ is the
inequality-constraint matrix. We next show that standard ADMM and primal-dual IPM
formulations differ in their numerical weights but not in the structural
sparsity of the reduced system.

\paragraph{IPM}
For the primal-dual IPM, after eliminating the slack variables, the regularized KKT matrix has the form
\begin{equation}
\label{eq:appendix_ipm_kkt}
K_{\mathrm{IPM}}=
\begin{bmatrix}
P+\tau I & A^\top & G^\top\\
A & -\delta I & 0\\
G & 0 & -(W+\delta I)
\end{bmatrix},
\end{equation}
where $\tau,\delta>0$ are regularization parameters and $W$ is a diagonal matrix with all positive entries. See~\cite{schwan2023piqp} for details.

\paragraph{ADMM}
Following the formulation of \texttt{OSQP}~\cite[Eq. (24)]{stellato2020osqp}, the linear system system to be solved for~\eqref{eq:appendix_qp} has the form
\begin{equation}
\label{eq:appendix_admm_kkt}
K_{\mathrm{ADMM}}=
\begin{bmatrix}
P+\sigma I & A^\top & G^\top\\
A & -\rho^{-1}I & 0\\
G & 0 & -\rho^{-1}I
\end{bmatrix},
\end{equation}
% \begin{equation}
% \label{eq:appendix_admm_kkt}
% K_{\mathrm{ADMM}}=
% \begin{bmatrix}
% P+\sigma I & \bar G^\top \\
% \bar G &  -\rho^{-1}I
% \end{bmatrix},
% \end{equation}
% where $\sigma$ and $\rho$ are positive parameters, and $\bar G \coloneqq [G^\top~ A^\top ~ -A^\top]^\top$ since we can view the $Ax=b$ as two-sided

We can see that matrices~\eqref{eq:appendix_ipm_kkt} and ~\eqref{eq:appendix_admm_kkt} have the same sparsity pattern. Solvers like \texttt{OSQP} directly factorize the matrix~\eqref{eq:appendix_admm_kkt}. Alternatively, we can work on a reduced system by eliminating the second and third block rows, which yields a linear system with a matrix
\begin{equation}
K=P+\tau I+\delta^{-1}A^\top A +G^\top(W+\delta I)^{-1}G,
\end{equation}
which is guaranteed to be symmetric positive definite.

% \section{Reduced KKT Structure of Branch MPC}
% \label{app:branch_kkt_structure}
% For problem~\eqref{eq:scenario_problem}, every scenario-dependent term involves only the shared root variable \(x_0\) and the tail variables \(x^i_{1:N}\) of one scenario. No objective or constraint term directly couples two distinct scenario tails. Consequently, after grouping the variables by scenario, the resulting matrix has a block-diagonal-arrow structure as in~\eqref{eq:root_coupled_matrix}. Moreover, the stage-local costs and constraints couple only neighboring temporal blocks along each tail, giving each tail matrix \(K_i\) a block-tridiagonal structure. Since \(x_0\) enters only the first stage of each tail, its coupling with \(K_i\) is confined to the single block \(C_{i,1}\) as in~\eqref{eq:permuted_tail_coupling}. 

\section{Block-sparsity in Tail Coupling Solve}
\label{app:proof_tail_coupling_solve}
\begin{theorem}
\label{thm:single_block_coupling_factor}
Consider the permuted coupling matrix \(\widehat C_i\)
in~\eqref{eq:permuted_tail_coupling}. Let \(L_i\) be the Cholesky
factor of \(\widehat K_i\), and let \(F_i\) be the unique solution of
\(F_iL_i^\top=\widehat C_i\). Then \(F_i\) has the same terminal-block
support as \(\widehat C_i\), namely,
$F_i=\begin{bmatrix}0&\cdots&0&F_{i,1}\end{bmatrix}$.
\end{theorem}

\begin{proof}
We write the matrix $L$ as
\begin{equation*}
L_i=
\begin{bmatrix}
\overline L_i & 0\\
G_i & L_{i,1}
\end{bmatrix},
\end{equation*}
where $L_{i,1}$ is the last diagonal block of $L_i$.
Writing \(F_i=[\overline F_i\;F_{i,1}]\) and transposing the coupling
solve gives \(L_iF_i^\top=\widehat C_i^\top\). By
\eqref{eq:permuted_tail_coupling}, this system is
\begin{equation*}
\begin{bmatrix}
\overline L_i & 0\\
G_i & L_{i,1}
\end{bmatrix}
\begin{bmatrix}
\overline F_i^\top\\
F_{i,1}^\top
\end{bmatrix}
=
\begin{bmatrix}
0\\
C_{i,1}^\top
\end{bmatrix}.
\end{equation*}
Since \(\overline L_i\) is nonsingular, the first block equation implies
\(\overline F_i=0\). Therefore, only the last block of \(F_i\) can be
nonzero.
\end{proof}

\end{document}